\documentclass[letterpaper, 10 pt, conference]{ieeeconf}  % Comment this line out if you need a4paper

\IEEEoverridecommandlockouts                              % This command is only needed if 
\title{\LARGE \bf
ReachLipBnB: A branch-and-bound method for reachability analysis of neural autonomous systems using Lipschitz bounds
}

\author{Taha Entesari,
        Sina Sharifi,
        Mahyar Fazlyab
        \thanks{Taha Entesari, Sina Sharifi, and Mahyar Fazlyab are with the department of  Electrical and Computer Engineering,
            Johns Hopkins University,
            {\tt\small \{tentesa1, sshari12, mahyarfazlyab\}@jhu.edu}}%
}

\usepackage[english]{babel}
\usepackage[utf8]{inputenc}

\usepackage{amsmath}
\usepackage{amsfonts}
\usepackage[dvipsnames]{xcolor}
\usepackage{graphicx}
\usepackage{subcaption}\usepackage{amsfonts}
\usepackage{bm}
\usepackage{bbm}
\usepackage{tikz}
\usepackage{epstopdf} %converting to PDF
\usepackage{mdframed}   % for framing

\usepackage{comment}
\usepackage{hyperref}
\usepackage[
backend=bibtex,
style=numeric-comp,
sorting=none,
maxbibnames=12,
]{biblatex}
\usepackage{float}
\usepackage{pdfpages}
\usepackage{algorithm} 
\usepackage[noend]{algpseudocode} 
\usepackage{multirow}
\usepackage{makecell}

\newlength\myindent
\newtheorem{theorem}{Theorem}

\newtheorem{proposition}{Proposition}

\newtheorem{remark}{Remark}

\begin{document}

\maketitle
\thispagestyle{empty}
\pagestyle{empty}

%%%%%%%%%%%%%%%%%%%%%%%%%%%%%%%%%%%%%%%%%%%%%%%%%%%%%%%%%%%%%%%%%%%%%%%%%%%%%%%%
\begin{abstract}
We propose a novel Branch-and-Bound method for reachability analysis of neural networks in both open-loop and closed-loop settings. Our idea is to first compute accurate bounds on the Lipschitz constant of the neural network in certain directions of interest offline using a convex program. We then use these bounds to obtain an instantaneous but conservative polyhedral approximation of the reachable set using Lipschitz continuity arguments. To reduce conservatism, we incorporate our bounding algorithm within a branching strategy to decrease the over-approximation error within an arbitrary accuracy. We then extend our method to reachability analysis of control systems with neural network controllers. Finally, to capture the shape of the reachable sets as accurately as possible, we use sample trajectories to inform the directions of the reachable set over-approximations using Principal Component Analysis (PCA). We evaluate the performance of the proposed method in several open-loop and closed-loop settings.

\end{abstract}

%%%%%%%%%%%%%%%%%%%%%%%%%%%%%%%%%%%%%%%%%%%%%%%%%%%%%%%%%%%%%%%%%%%%%%%%%%%%%%%%
\section{INTRODUCTION}

Going beyond machine learning tasks, neural networks also arise in a variety of control and robotics problems, where they function as feedback control policies \cite{berkenkamp2017safe,chen2022large,chen2018approximating}, motion planners \cite{qureshi2019motion}, perception modules/observers, or as models of dynamical systems \cite{ogunmolu2016nonlinear,dai2021lyapunov}. However, the adoption of these approaches in safety-critical domains (such as robots working alongside humans) has been hampered due to a lack of stability and safety guarantees, which can be largely attributed to the  large-scale and compositional structure of neural networks. These challenges only exacerbate when neural networks are integrated into feedback loops, in which time evolution adds another axis of complexity. 
%To overcome these limitations, a plethora of methods have been developed to formally verify neural networks, both in isolation and in closed-loop settings. 
%The goal of this paper is to develop an optimization-based framework for fast reachability analysis of neural network autonomous systems. 

Neural network verification and reachability analysis can often be cast as optimization problems of the form
\begin{align} \label{eq: verification as optimization}
J^\star := \mathrm{\min}_{x \in \mathcal{X}} \quad J(f(x)),
\end{align}
where $f$ is a neural network, $J$ is a function that encodes the property (or constraint) we would like to verify, and $\mathcal{X}$ is a bounded set of inputs. In this context, the goal is to either certify whether the optimal value $J^\star$ is greater than or equal to a certain threshold (zero without loss of generality), or to find a counter-example $x^\star $ that violates the constraint, $J^\star=J(x^\star) < 0$. When $J^\star \geq 0$, then the output reachable set $f(\mathcal{X})$ is contained in the zero super-level set of $J$, i.e., $f(\mathcal{X}) \subseteq \{y \mid J(y) \geq 0\}$. Thus, by choosing $J$ appropriately, we can localize the output reachable set. %For example when $J$ is a linear function, $J(y)=c^\top y$, then \eqref{eq: verification as optimization} corresponds to computing the support function of the reachable set $f(\mathcal{X})$. 

% In neural network verification, the goal is verify whether the neural network output satisfies certain constraints for a given set of inputs. This problem can be typically written as
% \begin{align} \label{eq: verification as optimization}
% J^\star := \mathrm{\min}_{x \in \mathcal{X}} \quad J(f(x)),
% \end{align}
% % \begin{align} \label{eq: verification as optimization}
% % J(f(x)) \geq 0 \quad \forall x \in \mathcal{X},
% % \end{align}
% %
% where $f$ is the neural network, $J$ is a function that encodes the constraint we would like to verify, and $\mathcal{X}$ is a bounded set of inputs. When \eqref{eq: verification as optimization} is satisfied, the reachable set $f(\mathcal{X})$ is contained in the zero sub-level set of $J$. Hence, a closely-related but more general problem is reachability analysis, where the goal is to over-approximate the reachable set $f(\mathcal{X})$ as accurately as possible.
%For verification of neural networks in isolation (e.g., data classification tasks), characterizing the exact reachable set is not necessary. However, when neural networks are used in closed-loop systems, estimating the reachable sets in each time step is required.  
\subsection{Related work}
For piece-wise linear networks and an affine $J$, \eqref{eq: verification as optimization} can be solved by SMT solvers \cite{katz2017reluplex}, or Mixed-Integer Linear Program (MILP) solvers \cite{tjeng2017evaluating,dutta2018output}, which rely on branch-and-bound (BnB) methods. To improve scalability and practical run-time, the state-of-the-art neural network verification algorithms rely on customized branch-and-bound methods, in which a convex relaxation of the problem is solved efficiently to obtain fast bounds \cite{anderson2020tightened,bunel2020lagrangian,de2021improved,xu2020fast,kouvaros2021towards,wang2021beta,ferrari2022complete}. 
In all these approaches
%For piece-wise linear neural networks and linear $J$, it is 
%State-of-the-art algorithms for neural network verification rely on customized Branch-and-Bound (BnB) methods \cite{anderson2020tightened,bunel2020lagrangian,de2021improved,xu2020fast,kouvaros2021towards,wang2021beta,ferrari2022complete}. The general idea is to divide the verification problem into subproblems (branching) that can be verified using convex relaxations, or specialized bound propagators using zonotopes, star sets, etc. 
branching can be done by splitting the input set directly, which can be efficient on low dimensional inputs when combined with effective heuristics \cite{wang2018formal,xiang2020reachable,anderson2020tightened,everett2020robustness}. {
\cite{rubies2019fast} proposes a heuristic that is used to decide which axis to split in the input space using shadow prices obtained from the dual problem.} 
%They propose a metric that uses first-order information to approximate the resulting bounds in case of splitting any of the possible axes.} 
For high dimensional inputs, splitting uncertain ReLUs\footnote{Rectified Linear Unit.} into being active or inactive can be more efficient. Stemming from this idea, \cite{vincent2021reachable} proposes to use the activation pattern of the ReLUs and enumerates the input space using polyhedra such that the neural network is an affine function in each subset. However, for large number of neurons or large inputs, these strategies can also become inefficient. In this paper, we focus on low dimensional inputs and large input sets, where splitting the input set directly can be more effective. 

%To avoid the worst-case complexity in practice, it is imperative to \emph{(i)} use effective heuristics to improve the branching strategy and avoid excessive branching; and \textit{(ii)} to develop custom solvers in the bounding part that can exploit the problem structure to compute the bounds efficiently.

\emph{Closed-loop verification.} Compared to open-loop settings, verification of closed-loop systems involving neural networks introduces a set of unique challenges. Importantly, the so-called wrapping effect \cite{neumaier1993wrapping} leads to compounding approximation error for long time horizons, making it challenging to obtain non-conservative guarantees efficiently over a longer time horizon and/or with large initial sets. 
%To mitigate this effect, it is essential to capture the shape of the reachable set at each time step as accurately as possible. 
%
Recently, there has been a growing body of work on reachability analysis of closed-loop systems involving neural networks \cite{huang2019reachnn,ivanov2019verisig,dutta2019reachability,everett2020certified,rober2022backward,tran2020nnv,julian2021reachability,sidrane2022overt,sun2019formal}. Most relevantly, \cite{hu2020reach} introduces a method called ReachSDP that abstracts the nonlinear activation functions of a neural network by quadratic constraints and solves the resulting semidefinite program to perform reachability analysis. \cite{everett2021efficient} proposes a sample-guided input partitioning scheme called Reach-LP to improve the bounds of the reachable set.

%In \cite{huang2019reachnn,fan2020reachnn} an approach based on Bernstein polynomial is used for reachability analysis. 
%\cite{dutta2019reachability} uses a polynomial model to approximate the dynamics of the controller.
%Other methods \cite{dutta2018learning} have proposed using some template polyhedra to calculate the output shape. This approach allows the algorithm to solve a maximization problems only in certain directions and could reduce the wrapping effect by producing a less conservative over-approximation at each step.
\subsection{Our contribution} 
In this paper, we propose a novel BnB method for the reachability analysis of neural network autonomous systems. As the starting point, we first consider polyhedral approximations of the reachable set of neural networks in isolation (problem \eqref{eq: verification as optimization} for linear $J$). Our idea is to compute accurate bounds on the Lipschitz constant of the neural network in specific directions of interest offline using a semidefinite program. We then use these Lipschitz bounds to obtain an instantaneous but conservative polyhedral approximation of the reachable set using Lipschitz continuity arguments. To reduce conservatism, we incorporate our bounding algorithm within a branching strategy to decrease the over-approximation error within an arbitrary accuracy. We then extend our method to reachability analysis of control systems with neural network controllers. To capture the shape of the reachable sets as accurately as possible, we use sample trajectories to inform the directions of the reachable set over-approximations using Principal Component Analysis (PCA).  In contrast to existing BnB approaches, our method does not require any bound propagation relying on forward and backward operations on the neural network.%To the best of our knowledge, this is the first work based on BnB to perform reachability analysis of neural network-controlled systems. 

%In this paper, we propose an efficient and accurate method to compute polyhedral over-approximation of the reachable sets of closed-loop systems with neural network controllers. Our idea is to first compute accurate bounds on the operator norm of the closed-loop map (i.e., its Lipschitz constant) in specific directions of interest offline using a convex program. We then use these computations to obtain fast but conservative polyhedral approximation of the reachable sets online using sensitivity arguments. To reduce conservatism, we incorporate our bounding algorithm within a branching strategy to decrease the over-approximation error within an arbitrary accuracy. Finally, to capture the shape of the reachable sets as accurately as possible, we use sample trajectories to inform the directions of the reachable set over-approximations using Principal Component Analysis (PCA). To the best of our knowledge, this is the first work based on branch-and-bound to perform reachability analysis of neural network-controlled systems. 

The paper is structured as follows. In $\S$ \ref{sec:networkVerification} we introduce the problem of interest, detail our proposed Lipschitz based BnB framework for finding $\varepsilon$-accurate bounds, and provide proof of convergence. In $\S$ \ref{sec:reachability} we extend the introduced framework to handle reachability analysis of neural autonomous systems. We employ principal component analysis (PCA), allowing the proposal of rotated rectangle reachable sets yet satisfying the requirements of the BnB framework, yielding tighter reachable sets overall. We finally provide experimental results in $\S$ \ref{sec:results}.
An implementation of this project can be found in the following repository
\href{https://github.com/o4lc/ReachLipBnB}{https://github.com/o4lc/ReachLipBnB}.

\subsection{Notation}
For a vector $x \in \mathbb{R}^n$, $x_i$ denotes the $i$-th element of the vector.
For vectors $x, y \in \mathbb{R}^n$, $x \leq y$ represents $n$ element-wise inequalities $x_i \leq y_i$.
For a symmetric matrix $A \in \mathbb{S}^{n}$, $A \preceq 0$ states that the matrix is negative semi-definite. $0_{n_1\times n_2}$ represents the matrix of size $n_1$ by $n_2$ with all zero elements, and $I_n$ represents the identity matrix of size $n$. For matrices $A^1, \cdots, A^n$ of arbitrary sizes, $\mathrm{blkdiag}(A^1, \cdots, A^n)$ is the block diagonal matrix formed with $A^1$ to $A^n$.

\section{A branch-and-bound method based on Lipschitz constants}\label{sec:networkVerification}
Consider the optimization problem \eqref{eq: verification as optimization} for linear objective functions and rectangular input sets,
\begin{equation}\label{eqn:generalOptimization}
    J_c^\star(\mathcal{X}):=\inf_{x \in \mathcal{X}} \{J_c(x):=c^\top f(x)\}, %\text{ where } \mathcal{X} = \lbrack l, u \rbrack.
\end{equation}
where $f \colon \mathbb{R}^{n_x} \to \mathbb{R}^{n_f}$ is a feed-forward neural network with arbitrary activation functions, and $\mathcal{X} = \lbrack l, u \rbrack := \{x \mid \ell \leq x\leq u\}$ is an $n_x$-dimensional input rectangle. We propose a branch-and-bound method to solve problem \eqref{eqn:generalOptimization} to arbitrary absolute accuracy, meaning that for any given $\varepsilon>0$, our algorithm produces provable upper and lower bounds
$\mathrm{BUB}$
and
$\mathrm{BLB}$, short for best upper bound and best lower bound, respectively,
%$\underline{J^*_c(\mathcal{X})}$
on $J_c^\star(\mathcal{X})$ such that 
\begin{equation}\label{eqn:bnbTerminationCondition}
\mathrm{BUB} - \mathrm{BLB} \leq \varepsilon.    
\end{equation}
The algorithm solves the problem by recursively partitioning the input rectangle $\mathcal{X}$ into disjoint sub-rectangles,   $\mathcal{X} = \bigcup_i \mathcal{X}_i$ (\texttt{Branch}). For each sub-rectangle, it computes the corresponding lower and upper bound (\texttt{Bound}) on $J_c^\star(\mathcal{X}_i)$. We then have 
\begin{equation}\label{eqn:optimalValueBounded}
    \min_{i = 1, \cdots, N} \underline{J_c}(\mathcal{X}^i) \leq 
    J_c^\star(\mathcal{X})\leq
    \min_{i = 1, \cdots, N} \overline{J_c}(\mathcal{X}^i).
\end{equation}
The left inequality follows from the fact that at least one partition includes the global minimum, implying that the lower bound for that partition is a lower bound on the global minimum. The right inequality follows from the fact that the optimal value of $J_c$ over \emph{any} partition is at least as large as the global minimum, implying that the best (minimum) upper bound over all partitions is a valid upper bound.

Overall, the algorithm starts from an initial value for $\mathrm{BUB}$ and $\mathrm{BLB}$ and iteratively improves them according to
\begin{align}\label{eqn:blbbubUpdateRule}
    \begin{split}
        \mathrm{BLB} \: &\leftarrow \max\{\mathrm{BLB},\min_{i = 1, \cdots, N} \underline{J_c}(\mathcal{X}^i)\}\\
    \mathrm{BUB} \: &\leftarrow \min \lbrace \mathrm{BUB}, \min_{i = 1, \cdots, N} \overline{J_c}(\mathcal{X}^i) \rbrace
    \end{split}
\end{align}
The algorithm eliminates those sub-rectangles that provably do not contain any solution to the original problem \eqref{eqn:generalOptimization} (\texttt{Prune}), and it terminates when \eqref{eqn:bnbTerminationCondition} is satisfied.
\ifx
representing the full (active) space. 
We have the following sub-routines:
\begin{enumerate}
    \item \texttt{Bound}: Takes as input a rectangle $\mathcal{X}^i$, and provides a lower bound $\underline{J_c}(\mathcal{X}^i)$ and and upper bound $\overline{J_c}(\mathcal{X}^i)$ on the value of $J_c^\star(\mathcal{X}^i)$. 
    
    \item \texttt{Branch}: Takes a rectangular partition $\lbrace \mathcal{X}^1, \cdots, \mathcal{X}^N \rbrace$ of our
    \textit{active} space and outputs a refined 
    partitioning $\lbrace \mathcal{X}^1, \cdots, \mathcal{X}^{N^\prime} \rbrace$, $N^\prime > N$. \texttt{Branch} chooses a subset of indices $Q$ and divides each $\mathcal{X}^i, i\in Q$ into sub-rectangles replacing $\mathcal{X}^i$ in our original partition.
    
    \item \texttt{Prune}: Takes as input a partition $\lbrace \mathcal{X}^1, \cdots, \mathcal{X}^N \rbrace$ of the \textit{active} space and deletes those rectangles that provably do not contain any solution of problem \eqref{eqn:generalOptimization}. %i.e.,
    %if $\underline{J_c}(\mathcal{X}^i) > \min_{1\leq j \leq N} \overline{J_c}(\mathcal{X}^i)$. 
    %If some subsets are removed, our \textit{active} space shrinks and we are left with the partition (after re-indexing) $\lbrace \mathcal{X}^1, \cdots, \mathcal{X}^{N^\prime} \rbrace$, $N^\prime < N$, of the active space.
\end{enumerate}
The algorithm calls \texttt{Bound} on the initial partitioning $\{\mathcal{X}^0\}$. If $\overline{J_c}(\mathcal{X}^0) - \underline{J_c}(\mathcal{X}^0) < \varepsilon$, the algorithm terminates. Otherwise, \texttt{Branch} is called, which yields a more refined partition $\lbrace \mathcal{X}^1, \cdots, \mathcal{X}^N \rbrace$. We then call  \texttt{Bound} on each sub-rectangle $\mathcal{X}^i, i=1,\cdots, N$. We then have 
\begin{equation}\label{eqn:optimalValueBounded}
    \min_{i = 1, \cdots, N} \underline{J_c}(\mathcal{X}^i) \leq 
    J_c^\star(\mathcal{X}^0)\leq
    \min_{i = 1, \cdots, N} \overline{J_c}(\mathcal{X}^i).
\end{equation}
The left inequality follows from the fact that at least one partition includes the global minimum, implying that the lower bound for that partition is a lower bound on the global minimum. The right inequality follows from the fact that the optimal value of $J_c$ over \emph{all} partitions is at least as large the global minimum, implying that the best upper bound over all partitions is a valid upper bound. We now set
\[
\mathrm{BLB} =\min_{i = 1, \cdots, N} \underline{J_c}(\mathcal{X}^i),\:\: \mathrm{BUB} = \min_{i = 1, \cdots, N} \overline{J_c}(\mathcal{X}^i).
\]

%The upper bound in \eqref{eqn:optimalValueBounded} just states that we pick the best upper bound of all sub-rectangles as an upper bound of the optimal value. The lower bound just states that we can guarantee with the worst lower bound on all sub-rectangles that we have until now.
If condition \eqref{eqn:bnbTerminationCondition} is satisfied,
% \begin{equation}\label{eqn:bnbTerminationCondition}
%     \min_{i = 1, \cdots, N} \overline{J_c}(\mathcal{X}^i) - \min_{i = 1, \cdots, N} \underline{J_c}(\mathcal{X}^i) \leq \varepsilon,
% \end{equation}
the algorithm terminates. Otherwise, \texttt{Prune} and \texttt{Branch} will be called, respectively. Then,  \texttt{Bound} is called on the parts of the partition that have been refined (as \texttt{Bound} has already been called on the parts of the partition that are left unchanged), and the bounds are updated as 
\begin{align}\label{eqn:blbbubUpdateRule}
    \begin{split}
        \mathrm{BLB} \: &\leftarrow \min_{i = 1, \cdots, N} \underline{J_c}(\mathcal{X}^i)\\
    \mathrm{BUB} \: &\leftarrow \min \lbrace \mathrm{BUB}, \min_{i = 1, \cdots, N} \overline{J_c}(\mathcal{X}^i) \rbrace.
    \end{split}
\end{align}
The algorithm loops (calling \texttt{Prune}, then \texttt{Branch}, and then \texttt{Bound} on the new subsets) until condition \eqref{eqn:bnbTerminationCondition} is satisfied. The output of the algorithm after termination is BLB, which is a lower bound on the optimal value of \eqref{eqn:generalOptimization}.
% \[
% \underline{J_c^\star(\mathcal{X}^0)} = \min_{i = 1, \cdots, N} \underline{J_c}(\mathcal{X}^i)
% \]

%where $\mathcal{X} = \lbrack l, u \rbrack$. %We will denote such sets as $\lbrack l, u \rbrack$

%The process of branching, as it will be explained in section \ref{subsec:branchingVerification}, can be seen as a tree graph. Each node $i$ is defined with $\lbrack l_i, u_i \rbrack$, and the lower and upper bounds that result from running the bounding sub-routine on $\lbrack l_i, u_i \rbrack$. The tree initially starts with a single node denoting the original $\mathcal{X}$. 
%We manually set the lower and upper bounds for this node as $\infty$ and $-\infty$ for consistency.
\fi
\subsection{Bounding}\label{subsec:verificationBounding}
The \texttt{Bound} sub-routine finds guaranteed lower and upper bounds on the minimum value of the optimization problem in \eqref{eqn:generalOptimization} for a given generic rectangle $\mathcal{X} = [\ell,u]$. We denote these bounds by $\underline{J_c}(\mathcal{X})$ and $\overline{J_c}(\mathcal{X})$, respectively:
\begin{align*}
\underline{J_c}(\mathcal{X}) \leq J_c^\star(\mathcal{X}) \leq \overline{J_c}(\mathcal{X}).
\end{align*}
For the upper bound, we can use any local optimization scheme such as Projected Gradient Descent (PGD) to obtain heuristically good upper bounds, which would require the computation of the (sub-)gradient of $J_c$. Furthermore, any feasible point $x \in \mathcal{X}$ provides a valid upper bound on the infimum value of $J_c$.  

In this paper, we use the center point $x_{\text{center}} = (\ell+u)/2$ to compute an upper bound, $\overline{J_c}(\mathcal{X}) = J_c(x_{\text{center}}).$
To obtain a lower bound, any convex relaxation can, in principle, be used. Our proposed \texttt{Bound} sub-routine uses Lipschitz continuity arguments to obtain instantaneous but more conservative bounds. Concretely, suppose $J_c$ is Lipschitz continuous on $\mathcal{X}$ in $\ell_2$ norm, i.e.,  
\begin{align*}
    |J_c(x)-J_c(y)| \leq L_{J_c} \|x-y\|_2 \quad \forall x,y \in \mathcal{X},
\end{align*}
where $L_{J_c}>0$ is a Lipschitz constant.  Using this inequality with $y=x_{\text{center}}$ we obtain
\[
J_c(x_{\text{center}}) -L_{J_c} \| x - x_{\text{center}} \|_2 \leq J_c(x).
\]
Taking the infimum over $\mathcal{X}$ yields
\[
J_c(x_{\text{center}}) - L_{J_c}\sup_{x \in \mathcal{X}} \| x - x_{\text{center}} \|_2 \leq \inf_{x \in \mathcal{X}} J_c(x).
\]
Since $\mathcal{X} = \lbrack l, u \rbrack$, we can upper bound the supremum as follows,
\[
\sup_{x \in \mathcal{X}} \| x - x_{\text{center}} \|_2 =\left(\sum_{i = 1}^n (\frac{u_i - l_i}{2})^2\right)^{\frac{1}{2}} = \frac{1}{2}\mathrm{diam}(\mathcal{X}).
\]
We finally have the desired lower bound as 

\begin{equation}\label{eqn:lowerBound}
    \underline{J_c}(\mathcal{X}) = J_c(x_{\text{center}}) - \frac{L_{J_c}}{2} \mathrm{diam}(\mathcal{X}). %\leq \inf_{x \in \mathcal{X}} J_c(x) \leq J_c(x). 
 \end{equation}
This lower bound is conservative, but its computation requires the calculation of a Lipschitz constant only once for the entire algorithm. More precisely, once a Lipschitz constant over a rectangle $\mathcal{X}$ (or a superset of it) is computed, it can be used for any subsequent partition of $\mathcal{X}$. 

\textit{Lower Bound Refinement:} 
We can utilize our quick bounding technique to improve the lower bounds on each partition. Given a rectangle $\mathcal{X}$, we split it into $k_v$ sub-rectangles $\mathcal{X} = \bigcup_{1 \leq i \leq k_v} \mathcal{X}_i$. Furthermore, suppose $\mathcal{X}_p \supset \mathcal{X}$ is an immediate parent rectangle of $\mathcal{X}$. Then we can write
\begin{equation*}
\underline{J_c}(\mathcal{X}) \!=\! \max\lbrace lb_1,lb_2,lb_3\rbrace.
\end{equation*}
where
\begin{align*}
    lb_1 & = J_c(x_{\text{center}}) - \frac{1}{2}L_{J_c} \mathrm{diam}(\mathcal{X})\\
    lb_2 &= \underline{J_c}(\mathcal{X}_p) \\
    lb_3 &= \min_{1 \leq i \leq k_v}  \underline{J_c}(\mathcal{X}_i)
\end{align*}
Here we are using the fact that the lower bound over $\mathcal{X}_p$ is also a valid lower bound over $\mathcal{X}$, and at least one $\mathcal{X}_i$ includes the global minimum over $\mathcal{X}$. In our implementation, we split the rectangle recursively along the longest edge(s) of the resulting sub-rectangles. This rule decreases the condition number of the sub-rectangles, resulting in a less conservative Lipschitz-based lower bound.

It now remains to compute a provable upper bound on the Lipschitz constant of $J_c$. In this paper, we use the framework of LipSDP from \cite{fazlyab2019efficient} to compute sharp upper bounds on the Lipschitz constant of $J_c$. To this end, we consider the following representation to describe an $L$-layer neural network,
\begin{equation} \label{eq: nn model}
% \begin{aligned}
%     g_0(x) &=x \\
%     g_{k+1}(x) &= \phi( W_k g_k(x)) \quad k=0,\cdots,\ell-1 \\
%     f(x) &=  W_{\ell} g_\ell(x).
% \end{aligned}
f(x) = (W^{L} \circ \phi \circ W^{L-1} \circ \cdots \phi \circ W^0 )(x),
\end{equation}
where $\lbrace W^i\rbrace_{i=0}^{L}$ is the sequence of weight matrices with $W^k \in \mathbb{R}^{n_{k+1}\times n_k}$, $\ k=0,\cdots,L$, and $n_0=n_x$, $n_{L+1}=n_f$. Here $\phi$ is the activation layer defined as $[\phi(x)]_i = \varphi(x_i)$, where $\varphi \colon \mathbb{R} \to \mathbb{R}$ is an activation function.  We have ignored the bias terms as the proposed method is agnostic to their values. We assume that the activation functions are slope-restricted, i.e., they satisfy $\alpha (x-y)^2 \leq (x-y)(\varphi(x)-\varphi(y)) \leq \beta (x-y)^2$, 
%
% \begin{align*}
%     \alpha (x-y)^2 \leq (x-y)(\varphi(x)-\varphi(y)) \leq \beta (x-y)^2,
% \end{align*}
%
for some $0 \leq \alpha<\beta<\infty$ and all $x,y \in \mathbb{R}$. 
We then note that $J_c$ is essentially a scalar-valued neural network whose weights are given by $\{W^0,\cdots,W^{L-1},c^\top W^L\}$. Therefore, we can directly use LipSDP to compute the Lipschitz constant of $J_c(x) = c^\top f(x)$. 

In \cite{hashemi2020certifying} the authors develop the local version of LipSDP in which the slope parameters $\alpha$ and $\beta$ are localized for each neuron based on a priori-calculated pre-activation bound. In our implementation, we use this version of LipSDP to localize the Lipschitz constant to the original input set.
%We now present a method to calculate $L_{J_c}$ for the task at hand \red{Thus far, everything has been $J_c(x)$ and we don't care about what $J_c(x)$ is. I guess we should decide where to introduce the control system as it can change the lipschitz calculations ....}

\subsection{Branch and Prune}\label{subsec:branchingVerification}
Given a rectangular partitioning of the active space, the branch sub-routine selects a subset of the partitions for refinement. To avoid excessive branching, it is imperative to use effective heuristics to select the most promising sub-rectangles. 
%Our experimental results suggest that using the lower bound heuristic is best. 
We propose to branch those sub-rectangles that have the lowest lower bounds given by \texttt{Bound} since they are more likely to contain the global minimum. Furthermore, choosing the lowest lower bound will always improve the global lower bound.

In our implementation, \texttt{Branch} sorts the sub-rectangles according to their score (the lower bounds) and chooses the first $k_b$ sub-rectangles to branch. Next, we explain which axis \texttt{Branch} splits.
There are many heuristics that can be employed in order to choose which axis to split, some of which are covered in \cite{bunel2018unified}.

%We will thus discard the other heuristics and only consider the lower bound heuristic.

%When the branching sub-routine is called, we sort  the nodes (descending in all methods except the last score heuristic) based on their score  and then choose the first $k_b$ nodes to branch. Next we'll describe how we branch a single node.
Given a generic rectangle 
$\mathcal{X} = \lbrack l, u \rbrack$, $l,u \!\in\! \mathbb{R}^n $, 
\texttt{Branch} chooses the axis $j$ of maximum length, i.e., $j = \arg\max_k (u_k - l_k)$.
The rationale behind choosing the longest edge is that based on (\ref{eqn:lowerBound}), splitting that edge results in a maximal reduction of $\mathrm{diam}(\mathcal{X})$ so it directly increases the lower bound.
\ifx
\texttt{Branch} then divides the chosen axis into $k_d$ equal parts. For example, if $j = 1$ then $\mathcal{X}$ is partitioned into $\lbrace \mathcal{X}^1, \cdots, \mathcal{X}^{k_d} \rbrace$, where $\mathcal{X}^m = \lbrack l^m, u^m \rbrack, m=1, \cdots, k_d$ and
\begin{equation*}
    \begin{aligned}
    l^{m} &= \begin{bmatrix}l_1 + (m - 1) \frac{u_1 \!-\! l_1}{k_d}, &
l_2, &
\cdots, &
l_n
\end{bmatrix}^\top\\ 
u^{m} &= \begin{bmatrix}l_1 + m \frac{u_1 - l_1}{k_d}, &
u_2, &
\cdots, &
u_n
\end{bmatrix}^\top.
\end{aligned}
\end{equation*}
%The algorithm now outputs (and calls the bounding sub-routine on) the aforementioned sets $\mathcal{X}_m$. The parent node is no longer needed and is deleted from memory.\\
\texttt{Branch} replaces $\mathcal{X}$, the parent rectangle, with $\lbrace \mathcal{X}^1, \cdots, \mathcal{X}^{k_d} \rbrace$, the children sub-rectangles, in the current \textit{active} space partitioning.
\fi

%Branching is the action of splitting the input space into a finite number of (preferably) mutually exclusive sub-spaces in the hope that the smaller sub-problems are easier(\red{?}) to solve. In order for this splitting action to be done, a two questions need to be answered; "Which node should be split first?" and "How should we split the chosen node?".
%Regarding these questions, many different heuristics have been suggested, but in this work we mainly choose the worst node in sense of volume or length of edge to split first and we split them into k mutually exclusive identical sub-spaces.

% The pruning sub-routine plays an important role
% In the course of the algorithm, we keep track of the best lower bound (BLB) and the best upper bound (BUB). The update rules for these values are
% \begin{equation}
%     \text{BLB} = \min_{\text{all alive nodes }i} lb^i 
% \end{equation}

% \begin{equation}
%     \text{BUB} = \min \lbrace \text{current BUB}, \min_{\text{all alive nodes }i} ub^i \rbrace 
% \end{equation}

% If at any stage the lower bound of a node is greater than the BUB, that node is pruned and is not explored further (and deleted from memory).\\
%The algorithm for solving problem \eqref{eqn:generalOptimization} will terminate when $ \text{BUB} - \text{BLB} \leq \varepsilon $ for some choice of user specified $\varepsilon$

Finally, given a partitioning of the active space, the \texttt{Prune} sub-routine deletes those rectangles that cannot contain any global solution of the original optimization \eqref{eqn:generalOptimization}. To do this, \texttt{Prune} simply discards rectangles $\mathcal{X}^i$ for which
% \begin{equation*}
$\underline{J_c}(\mathcal{X}^i) > \mathrm{BUB},$
% \end{equation*}
%
yielding a smaller \textit{active} partition. This shows the importance of estimating a good lower bound, in other words by having a good estimate of the lower bound, we could reduce our search space. That is why \textit{Bound Refinement} can empirically reduce the total number of branches which is shown in Table \ref{tab:virtualBranching}.

\subsection{Convergence}
To show convergence, we simply show that our proposed framework fits the branch-and-bound framework described in \cite{boyd2007branch} and satisfies the sufficient conditions for convergence. Our \texttt{Bound} sub-routine provides guaranteed upper and lower bounds on the value of the objective function on any rectangle. This is exactly the condition (R1) in \cite{boyd2007branch}.
%Following the arguments in $\S$\ref{subsec:verificationBounding}, we have that $\underline{J_c}(\mathcal{X}) \leq \inf_{x\in \mathcal{X}} J_c(x)$, and as the upper bound calculated for each sub-rectangle $\mathcal{X}$ is simply the value of $J_c$ at a point in $\mathcal{X}$, thus 
% $
% \underline{J_c}(\mathcal{X}) \leq \inf_{x\in \mathcal{X}} J_c(x) \leq 
% \overline{J_c}(\mathcal{X}).
% $
% This is condition (R1) in \cite{boyd2007branch}.\\
For any given $\varepsilon > 0$, let $\delta = \frac{2\varepsilon}{L_{J_c}}$. Then for any $\mathcal{X}$ with $\mathrm{diam}(\mathcal{X}) \leq \delta$ we have $\overline{J_c}(\mathcal{X}) - \underline{J_c}(\mathcal{X}) \leq\varepsilon$. This is equivalent to condition (R2) of \cite{boyd2007branch}.\\
% As $J_c$ is a continuous function, for any given $\varepsilon > 0$, there exists a $\delta > 0$ such that for all $x^\prime$ that $\mid\mid x^\prime - x\mid\mid_2 \leq \delta$, $\mid J_c(x) - J_c(x^\prime)\mid \leq \varepsilon$. Consider a given $\varepsilon$ and its corresponding $\delta$. Then, for any sub-rectangle $\mathcal{X}$ with $\text{diam}(\mathcal{X}) \leq 2\delta$ we have $\mid J_c(x) - J_c(x_{\text{center}})\mid \leq \varepsilon$ for all $x \in \mathcal{X}$. As a result, 
% \[
% \mid \overline{J_c}(\mathcal{X}) - J_c(x_{\text{center}})\mid \leq \varepsilon
% \]
% Now, if $L\delta > \varepsilon$, choose the largest $\delta^\prime$ that $L\delta^\prime \leq \varepsilon$ and then set $\delta = \delta^\prime$. Then, by equation \eqref{eqn:lowerBound}, we have
% \[
% J_c(x_{\text{center}}) - \underline{J_c}(\mathcal{X}) \leq \varepsilon
% \]
% We thus conclude that $\overline{J_c}(\mathcal{X}) - \underline{J_c}(\mathcal{X}) \leq 2\varepsilon$. This is equivalent to condition (R2) of \cite{boyd2007branch}.\\
As conditions (R1) and (R2) of \cite{boyd2007branch} are satisfied for the bounding sub-routine and we are using the same branching sub-routine, the convergence of our algorithm follows. 
\subsection{Termination}
The procedure explained so far in section \ref{sec:networkVerification} is an algorithm for globally solving problem \eqref{eqn:generalOptimization} within an arbitrary absolute accuracy. However, for the task of neural network verification, there are two additional criteria that can help in the early termination of the algorithm. Recall that in neural network verification, the goal is to check whether $J_c^\star(\mathcal{X})  \geq 0$ holds. 
Thus, if at any iteration of the algorithm, we find $\mathrm{BLB} \geq 0$, we can terminate the algorithm and the problem is verified. Similarly, if we find that $\mathrm{BUB} < 0$, then we have found a counterexample for our verification task and we can also terminate the algorithm.

\section{Reachability Analysis of Neural Autonomous Systems}\label{sec:reachability}

\subsection{Problem statement}
In this section, we extend our method to reachability analysis of linear dynamical systems with neural network controllers. Consider the discrete-time linear system
\begin{equation}\label{eq:LTI system}
    x^{t+1} = A^t x^t + B^t u^t.
\end{equation}
where $x^t \in \mathbb{R}^{n_x}$ is the state at time index $t=0,1\cdots$, $A \in \mathbb{R}^{n_x \times n_x}$ is the state matrix, and $B \in \mathbb{R}^{n_x \times n_u}$ is the input matrix. We assume that the feedback control policy is given by $u^t = f(x^t)$, where $f$ is a feed-forward neural network, giving rise to the following closed-loop autonomous system,
\vspace{-2mm}
\begin{equation}\label{eq:neural autonomous system}
     x^{t+1} = F^t(x^t):= A^t x^t + B^t f(x^t). 
\end{equation}
Given a set of initial conditions $\mathcal{X}^0 \subset \mathbb{R}^{n_x}$, a goal set $\mathcal{G} \subset \mathbb{R}^{n_x}$, and a sequence of avoid sets $\mathcal{A}^t \subset \mathbb{R}^{n_x}$, our goal is to  verify if all initial states $x^0 \in \mathcal{X}^0$  can reach $\mathcal{G}$ in a finite time horizon $T \geq 1$, while avoiding $\mathcal{A}^t$ for all $t = 0,\cdots,T$. To this end, we must compute the reachable sets defined as $\mathcal{X}^{t+1} = F^t(\mathcal{X}^t)$, and then verify that
\vspace{-2mm}
\begin{equation*}
    \mathcal{X}^t \cap \mathcal{A}^t = \emptyset \quad t=0,1,\cdots,T \quad \text{and} \quad \mathcal{X}^{T} \subseteq \mathcal{G}.
\end{equation*}
Since computing the exact reachable sets is computationally prohibitive even for simple dynamical systems \cite{fijalkow2019decidability}, almost all approaches overapproximate the true reachable sets \emph{recursively} using a set representation such as polytopes, ellipsoids, etc;  that is, assuming that at time $t$ the reachable set $\mathcal{X}^t$ has been over-approximated by $\overline{\mathcal{X}}^t$, we then over-approximate the reachable set $F^t(\overline{\mathcal{X}}^t)$ by $\overline{\mathcal{X}}^{t+1}$. Two critical questions are what set representation lends itself to the proposed BnB method in $\S$\ref{sec:networkVerification}; and how we can compute the corresponding Lipschitz constants. We address these two questions in the next subsection. 

    \subsection{Recursive reachability analysis via PCA-guided rotated rectangles}
\begin{figure}[t]
    \centering
    \vspace{1em}
    \includegraphics[width=.9\columnwidth]{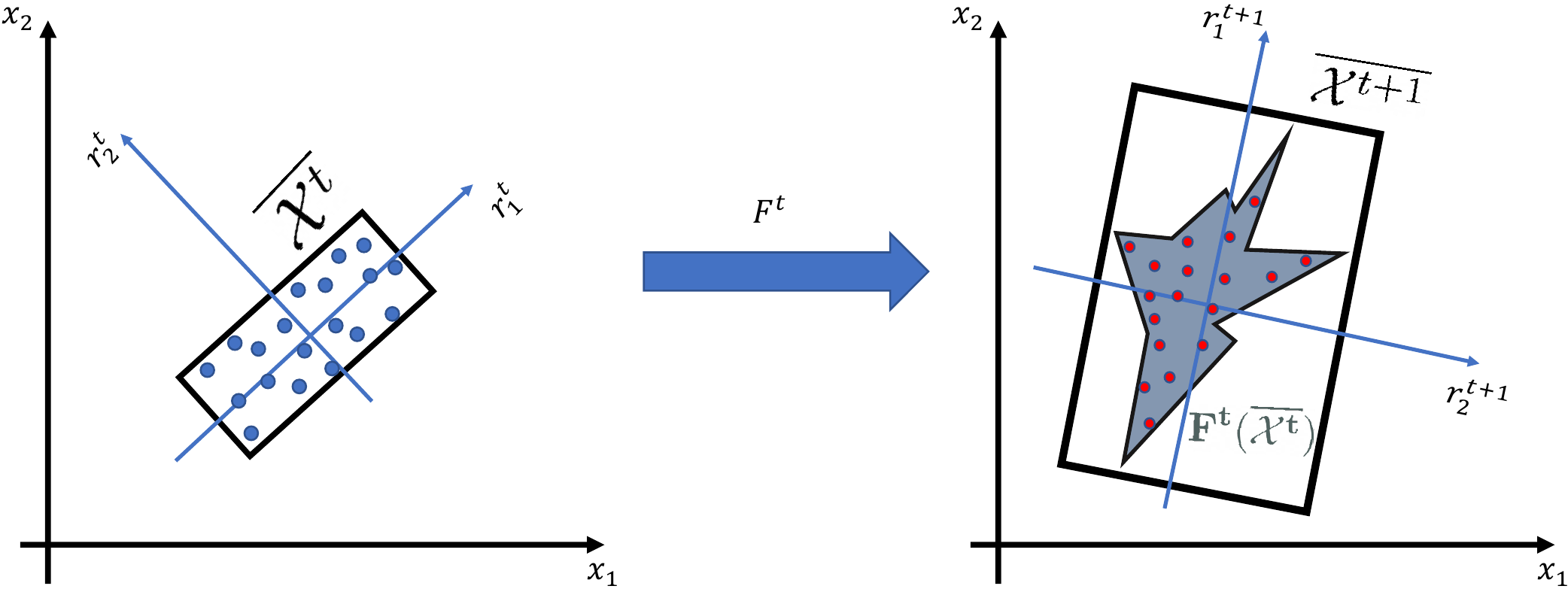}
    \caption{\small Using PCA to estimate the orientation of rectangles.}
    \label{fig:rotation}
\end{figure}
If we choose axis-aligned rectangles to over-approximate the reachable sets, then we can readily compute these sets recursively using the proposed BnB method in $\S$\ref{sec:networkVerification}.
Suppose $\overline{\mathcal{X}}^t = [\ell^t,u^t]$ has been computed. To compute $\overline{\mathcal{X}}^{t+1} = [\ell^{t+1},u^{t+1}]$, we must solve the optimization problems $\ell_i^{t+1} = \min_{x \in \overline{\mathcal{X}^t}} \{{e^i}^\top F^t(x)\}, \quad  
    u_i^{t+1} = -\min_{x \in \overline{\mathcal{X}^t}} \{-{e^i}^\top F^t(x)\},$
%
% \begin{align} \label{eq: verification closed loop}
%     J_c^\star(\overline{\mathcal{X}^t}) = \min_{x \in \overline{\mathcal{X}^t}} \{c^\top F^t(x)\}
% \end{align}
% \begin{align*}
%     \ell_i^{t+1} = \min_{x \in \overline{\mathcal{X}^t}} \{{e^i}^\top F^t(x)\}, \quad  
%     u_i^{t+1} = -\min_{x \in \overline{\mathcal{X}^t}} \{-{e^i}^\top F^t(x)\},
%     % \mathcal{X}^{t+1} \subseteq \overline{\mathcal{X}}^{t+1} =  \bigcap_{c \in \{\pm e^i\}} \{x \mid c^\top x \geq \underline{J_c^\star(\overline{\mathcal{X}}^t)}\}
% \end{align*}
%
where $e^i$ is $i$-th  standard basis vector.
Axis-aligned rectangles, however, might be too conservative as they might not capture the shape of the true reachable sets accurately. To mitigate this conservatism, we propose to use rotated rectangles defined as
%
%\begin{align*}
    $
    \overline{\mathcal{X}}^t = \{x \mid x = R^t y, \ y \in [\ell^t,u^t]  \}$
%\end{align*}
where $R^t$ is an orthonormal change-of-basis matrix from $x$ to $y$, and $\ell^t, u^t$ are lower and upper bounds on $y$. 
This set representation can be incorporated into our BnB framework by simply branching in the $y$ space. Now it remains to determine $R^t$, $\ell^t$, and $u^t$.

Inspired by \cite{stursberg2003efficient}, we derive the rotation matrix $R^{t}$ at each time $t$ from sample trajectories. More precisely, suppose we have $p$ sample trajectories $\{x^{0,i},\cdots,x^{T,i}\}_{i=1}^{p}$. Then by applying PCA on the data points $\{x^{t,i}\}_{i=1}^{p}$, we obtain an orthonormal basis $R^{t}$ in which the sampled points are uncorrelated. Thus, we have $R^{t} = [r^{t, 1},\cdots,r^{t, n_x}]$, where $r^{t, i}$ are the principal directions.

To determine $\ell^t, u^t$, the following proposition shows that we can find these bounds recursively. 

\begin{proposition}\label{thm:boundsOnNewCoordinates}\normalfont
For $i=1,2,\cdots,n_x$ define 
\begin{align} \label{eq: bounds on new coordinates}
    %J_c^\star(\bar{\mathcal{X}^t}) &= \min_{y \in [\ell^t,u^t]} \{J_c(y) = c^\top F(R^t y + b^t)\}
    \begin{split}
    \ell_i^{t+1} &= \min_{y \in [\ell^t,u^t]} \{{r^{t+1, i}}^\top F^t(R^t y)\} \\ 
    u_i^{t+1} &= -\min_{y \in [\ell^t,u^t]} \{{-r^{t+1, i}}^\top F^t(R^t y)\}.
    \end{split}
\end{align}
%
%
% for $c \in \{\pm r_i^{t+1}\}$ to obtain the oriented rectangle at time $t+1$,
% \begin{align}
%     \bar{\mathcal{X}}_{t+1} =  \bigcap_{c \in \{ r_i^{t+1}\}} \{x \mid  \underline{J_{c}^\star(\bar{\mathcal{X}}^t)} \leq c^\top x \leq - \underline{J_{-c}^\star(\bar{\mathcal{X}}^t)}\}
% \end{align}
% Thus by defining $y={R^t}^\top x$ and 
% \[
% l^{t+1}_i = \underline{J^\star_{c_i}(\bar{\mathcal{X}}^t)},
% \quad
% u^{t+1}_i = -\underline{J^\star_{-c_i}(\bar{\mathcal{X}}^t)}.
% \]
Then, we have
\[
F(\overline{\mathcal{X}}^t) \subseteq \overline{\mathcal{X}}^{t+1} =  \{x \mid x = {R^{t+1}} y, \ y \in [\ell^{t+1},u^{t+1}]\}.
\]
\end{proposition}
\begin{proof}
Let $x \in F(\overline{\mathcal{X}}^t)$. As $R^t$ is orthonormal, there exists a unique $y$ such that $x = R^{t + 1}y$, so
\[
y = {R^{t + 1}}^\top x = \begin{bmatrix}
{r^{t + 1, 1}}^\top x, &
\cdots,&
{r^{t + 1, n_x}}^\top x
\end{bmatrix}^\top.
\]
By definition in \eqref{eq: bounds on new coordinates}, we have 
\[
\ell^{t + 1}_i =
\min_{y \in [\ell^t,u^t]} \{{r^{t+1, i}}^\top F^t(R^t y)\} \leq 
{r^{t + 1, i}}^\top x = y_i,
\]
and
\begin{align*}
y_i& \leq 
-\!\min_{y \in [\ell^t,u^t]} \{{-r^{t+1, i}}^\top F^t(R^t y)\} \!=\! u_i^{t + 1}. 
\end{align*}
% \begin{align*}
%     \begin{split}
%         y_i &= {r^{t + 1, i}}^\top x\! \leq \!
% \max_{y \in [\ell^t,u^t]} \{{r^{t+1, i}}^\top F^t(R^t y)\}\! \\
% \rightarrow y_i& \leq 
% -\!\min_{y \in [\ell^t,u^t]} \{{-r^{t+1, i}}^\top F^t(R^t y)\} \!=\! u_i^{t + 1}. 
%     \end{split}
% \end{align*}
\end{proof}

According to Proposition \ref{thm:boundsOnNewCoordinates}, to over-approximate the reachable set at time $t+1$ using rotated rectangles, we must solve the optimization problems in \eqref{eq: bounds on new coordinates}. We can solve these problems using the BnB framework of $\S$\ref{sec:networkVerification} provided that we can bound the Lipschitz constant of 
\begin{align}
    J_c(y) = c^\top (A^t R^t y + B^t f(R^t y)).
\end{align}
We propose an extension of LipSDP that computes guaranteed bounds on the Lipschitz constant of $J_c$. To state the result, we define
\begin{align} \label{eq: A and B matrices}
\begin{split}
    A_{F} &= \begin{bmatrix}
\mathrm{blkdiag}(W^0R^t,\cdots,W^{L-1}) & 0_{N \times n_L}
\end{bmatrix}, \\
B_{F} &= \begin{bmatrix} 0_{N \times n_0} & I_{N}
\end{bmatrix}, \\
C_{F} &= \begin{bmatrix}
c^\top A^t R^t & 0_{1 \times (N-n_L)} & c^\top B^t W^L
\end{bmatrix}, \\ 
D_F &=\begin{bmatrix}
I_{n_0} & 0_{n_0 \times N} 
\end{bmatrix},
\end{split}
\end{align}
where $N = \sum_{i=1}^{L} n_i$ denotes the number of neurons. 
%
%

% \begin{figure}[t]
%     \centering
%     \includegraphics[width=.8\columnwidth]{pictures/RobotArm.png}
%     \caption{\small Polyhedral approximation of the reachable set for the robotic arm example.}
%     \label{fig:RobotArm}
% \end{figure}

% \begin{figure}[t]
%     \centering
%     \begin{subfigure}[b]{0.48\columnwidth}
%          \includegraphics[width=0.85\columnwidth]{pictures/BnB_Sequence1.png}
%         \caption{$c = [-0.994, -0.104]^\top$}
%         % \label{fig:Space Partition1}
%     \end{subfigure}%
%     \begin{subfigure}[b]{0.48\columnwidth}
%         \includegraphics[width=0.85\columnwidth]{pictures/BnB_Sequence2.png}
%         \caption{$c = [0.994, 0.104]^\top$}
%         % \label{fig:Space Partition2}
%     \end{subfigure}
    
%     \caption{\small Input space partitions made by the branch sub-routine on two directions for the robotic arm task.}
%     \label{fig:Space Partition}
% \end{figure}

\begin{figure}[t]
    \centering
    \vspace{2mm}
    \includegraphics[width=0.95\columnwidth]{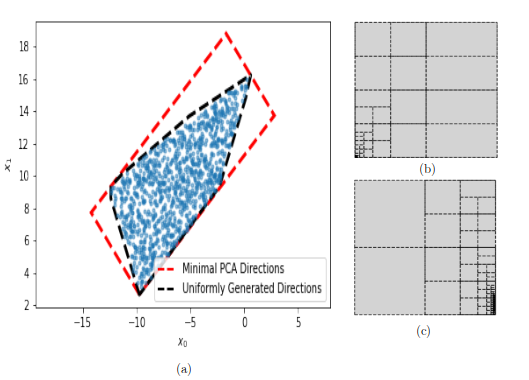}
    \caption{(a) \small Polyhedral approximation of the reachable set for the robotic arm example. (b-c) \small Input space partitions made by the branch sub-routine on two directions ($c = [-0.994, -0.104]^\top$ and $c = [0.994, 0.104]^\top$) for the robotic arm task.}
    \label{fig:RobotArm}
\end{figure}

\begin{theorem}
	\label{thm:closed-loop lip -1}
	\normalfont
	Consider an $L$-layer fully connected neural network described by \eqref{eq: nn model}. Suppose the activation functions are slope-restricted in the sector $[\alpha,\beta]$. Define $A_F,B_F,C_F,D_F$ as in \eqref{eq: A and B matrices}. Define the matrix
	\begin{align} \label{thm:closed-loop lip 0}
	\begin{bmatrix}
	A_F \\ B_F
	\end{bmatrix}^\top \! \begin{bmatrix}
	-2\alpha \beta T & (\alpha\!+\!\beta)T \\ (\alpha\!+\!\beta)T & -2T
	\end{bmatrix} \begin{bmatrix}
	A_F \\ B_F
	\end{bmatrix} \!+\! C_F^\top C_F \!-\! \rho D_F^\top D_F,
	\end{align}
	%holds for some $T \in \mathcal{T}_{n}$.
	%where $T \in \mathcal{T}_n$
	% 		, and $M$ is given by
	% 		\begin{align} \label{thm:multi-layer-thm 3}
	% 		M = \begin{bmatrix}
	% 		-L_2^2 I & 0 & \hdots & 0 \\
	% 		0 & 0 & \hdots & 0 \\
	% 		\vdots & \vdots & \ddots & \vdots \\
	% 		0 & 0 & \hdots & (W^{\ell})^\top W^{\ell} 
	% 		\end{bmatrix}.
	%\end{align}
	where $\rho>0$ and $T $ is a diagonal positive semidefinite matrix of size $n \times n$. If \eqref{thm:closed-loop lip 0} is negative semidefinite for some $(\rho,T)$, then $\sqrt{\rho}$ is a Lipschitz constant of $J_c(y)=c^\top (A^tR^t y + B^t f(R^t y))$ in $\ell_2$ norm.
	%\end{align*}
	%
\end{theorem}
\begin{proof}
Starting from the compositional representation of the neural network in \eqref{eq: nn model}, we can write $J_c$ as 
\begin{equation*}
\begin{aligned}
     \xi^0 &= y, \quad  \xi^{1} = \phi(W^0 R^t \xi^0), \\
     \xi^{k+1} &=\phi(W^k \xi^k), \quad k=1,\cdots,L-1, \\
     J_c(y) &= c^\top A^t R^t \xi^0 + c^\top B^t W^L \xi^L.
\end{aligned}
\end{equation*}
Following \cite{fazlyab2019efficient}, we can compactly write these equations as
\begin{equation} \label{thm:closed-loop lip 0.5}
    \begin{aligned}
         B_F \xi = \phi(A_F \xi), \quad J_c(y) = C_F \xi,
    \end{aligned}
\end{equation}
where $\xi = [{\xi^0}^\top \cdots {\xi^{L}}^\top]^\top$. Define the matrices
\begin{align*}
    M^1 &= \begin{bmatrix}
	A_F \\ B_F
	\end{bmatrix}^\top \begin{bmatrix}
	-2\alpha \beta T & (\alpha\!+\!\beta)T \\ (\alpha\!+\!\beta)T & -2T
	\end{bmatrix} \begin{bmatrix}
	A_F \\ B_F
	\end{bmatrix}, \\ \notag
    M^2 &= C_F^\top C_F \!-\! \rho D_F^\top D_F.
\end{align*}
Suppose $\overline{y}$ and $\overline{\xi}$ satisfy \eqref{thm:closed-loop lip 0.5}. Left and right multiply $M^1$ by $(\xi-\overline{\xi})^\top$ and $(\xi-\overline{\xi})$, respectively to obtain

\begin{algorithm}[b!]

	\caption{ReachLipBnB} 
	\textbf{Inputs}: Set of initial states $\mathcal{X}^0$, Neural network $f$, System dynamics $A^t, B^t$, Final horizon $h_f$, Number of trajectories $p$\\
	\textbf{Outputs}: Reachable sets $\mathcal{Y}^i$ and corresponding rotation transformations $R^i$ for $i=1, \cdots, h_f$\\
	\textbf{Initialize}:
	$\mathcal{Y}^0 = \mathcal{X}^0$, $R^0 = I_{n_x}$, Trajectories $ \lbrace x^{t, j}\rbrace _{t=0}^{h_f}, j=1,\cdots, p$.
	\begin{algorithmic}[1]
		\For {t $=0,1,\cdots, h_f - 1$}
		    \State partitions $= \{\mathcal{Y}^t\}$
		    \State directions = \texttt{PCA($\lbrace x^{t + 1, j} \rbrace_{j=1}^p$)}
		    \State $F^t(\cdot) = A^tR^t(\cdot) + B^t f(R^t(\cdot))$
			\For {$c$  in directions}
			\State $J = c^\top F^t$
			\State $L_J = \mathrm{LipSDP}(J)$
			\State $\mathrm{BUB} = \min_{j=1, \cdots, p} J(x^{t + 1, j})$,  $\mathrm{BLB} = -\infty$
			    \While {$\mathrm{BUB} - \mathrm{BLB} > \varepsilon$}
				    % \State Solve $J_c^\star(\mathcal{X}):=\inf_{x \in \mathcal{X}} \{J_c(x):=c^\top f(x)\}$
				    \State 
				    % $\lbrace 
				    % \begin{matrix}
				    % \overline{J_c}(\mathcal{X}^i)\\ \underline{J_c}(\mathcal{X}^i)
				    % \end{matrix}
				    % \rbrace_{\mathcal{X}^i \in \text{Partitions}} =$ 
				   \texttt{Bound$_{J, L_J}$}(partitions)
				    \State $\mathrm{BLB}, \mathrm{BUB} \leftarrow$ \eqref{eqn:blbbubUpdateRule}
				    \State partitions = \texttt{Prune}(partitions, $\mathrm{BUB}$)
				    \State partitions = \texttt{Branch}(partitions)
				    
				    % \State \begin{small} Partitions = \texttt{Branch(Prune(Partitions))}
				    % \end{small}
				\EndWhile
			\EndFor
		\State $\mathcal{Y}^{t + 1} \leftarrow $ using \eqref{eq: bounds on new coordinates}
		\State $R^{t + 1} \leftarrow $ directions
		\EndFor
	\end{algorithmic} 
\end{algorithm}

\begin{align} \label{thm:closed-loop lip 1}
&(\xi-\overline{\xi})^\top M^1 (\xi-\overline{\xi}) = \\ \notag
&\begin{bmatrix}
A_F (\xi-\overline{\xi}) \\ \phi(A_F\xi)-\phi(A_F\overline{\xi}))
\end{bmatrix}^\top \begin{bmatrix}
-2\alpha \beta T & (\alpha+\beta)T \\ (\alpha+\beta)T & -2T
\end{bmatrix} \begin{bmatrix}
\star
\end{bmatrix} \geq 0
\end{align}
where the last inequality follows from the fact that $\phi$ is slope restricted \cite{fazlyab2019efficient}. Similarly, we can write
\begin{align} \label{thm:closed-loop lip 2}
(\xi-\overline{\xi})^\top M^2 (\xi-\overline{\xi}) = (J(y)-J(\overline{y}))^2- \rho \|y-\overline{y}\|_2^2,
\end{align}
By adding both sides of \eqref{thm:closed-loop lip 1} and \eqref{thm:closed-loop lip 2}, we get
\begin{align*} %\label{thm:multi-layer-thm 7}
(\xi-\overline{\xi})^\top (M^1+M^2) (\xi-\overline{\xi}) \geq 
(J(y)-J(\overline{y}))^2 \!-\! \rho \|y-\overline{y}\|_2^2
\end{align*}
When the LMI in \eqref{thm:closed-loop lip 0} is feasible, the left-hand side is non-positive, implying that the right-hand side is non-positive. Thus, $|J(y)-J(\overline{y})| \leq \sqrt{\rho} \|y-\overline{y}\|_2$. 
\end{proof}

% \begin{remark}
% The acquired Lipschitz constant is for input perturbations in the $\ell_2$ sense. However, the perturbations considered in the field and in this paper are in the $\ell_\infty$ norm. To handle this, we simply employ the inequality $\| x \|_2 \leq \sqrt{n} \| x \|_\infty$ for $x\in \mathbb{R}^n$.
% \end{remark}
\begin{remark}
The trajectories used in the PCA can also provide us with a good initialization for $\mathrm{BUB}$ by choosing the maximum of the objective values evaluated at these points.
\end{remark}

\section{NUMERICAL RESULTS}\label{sec:results}

In this section, we evaluate and compare our algorithm on three tasks with those of Reach-LP \cite{everett2021efficient}\footnote{For comparison of results with \cite{everett2021efficient} we used their GitHub repository and extracted the relevant codes.} and Reach-SDP \cite{hu2020reach}. 
The first problem is a 2-dimensional open-loop task and the rest are closed-loop reachability tasks.
The experiments are conducted on an Intel Xeon W-2245 3.9 GHz processor with 64 GB of RAM. Throughout the experiments, we let $k_b$ be 512.

\subsection{Robotic Arm}
The first experiment is similar to the robotic arm test case from \cite{everett2020robustness}. For this experiment, $10^4$ data points were generated and used to train a single-layer neural network with two inputs, two outputs, and 50 hidden neurons.
The network takes $\theta^1$ and $\theta^2$ as input and maps it to all the possible points in $x-y$ plane that the robot can reach.
The input constraint is $\frac{\pi}{3} \leq \theta^1, \theta^2 \leq \frac{2 \pi}{3}$.\\
Fig. \ref{fig:RobotArm} shows the results on this problem in two cases; \emph{(i)} using PCA to generate the four directions and \emph{(ii)} using 60 uniformly spaced vectors between $[0, 2\pi]$. Fig. \ref{fig:RobotArm} also shows the resulting input space partitioning for two directions.
% \begin{table}[b!]
% \setlength{\tabcolsep}{2pt}
% \centering
%     \begin{tabular}{|c | c c c | c| c|} 
%      \hline
%      && ReachLipBnB & &  ReachLP & ReachSDP \\ [0.5ex] 
%      \hline
%      $\varepsilon$ & 0.1 & 0.01 & 0.001 & - & - \\ 
%      \hline
%      Time[s] & 1.06 $\pm$ 0.07  & 2.23 $\pm$ 0.08 & 4.78 $\pm$ 0.12 & 0.71 & 0.99 \\ 
%      % \hline
%      % $\sigma^2$[s$^2$]  & 0.001 & 0.004 & 0.33 & - & - \\ 
%      \hline
%     \end{tabular}
% \caption{\small Time comparison for different sub-optimality tolerances for the double integrator experiment. The Reach-LP algorithm is run with GSG as the partitioner and CROWN as the propagator.}
% \label{tab:double integrator}
% \end{table}

\begin{table}[b!]
\setlength{\tabcolsep}{1pt}
\centering
    \begin{tabular}{|c | c | c | c  c c|} 
     \hline
     &ReachLP& ReachSDP& \multicolumn{3}{c|}{ReachLipBnB}     \\ [0.5ex] 
     \hline
     $\varepsilon$ & - & - & 0.1 & 0.01 & 0.001 \\ 
     \hline
     Time[s] & 0.89 & 2.01 & 0.989 $\pm$ 0.058  & 1.204 $\pm$ 0.055 & 1.518 $\pm$ 0.066 \\ 
     % \hline
     % $\sigma^2$[s$^2$]  & 0.001 & 0.004 & 0.33 & - & - \\ 
     \hline
     % \multicolumn{3}{|c|}{Num total branches} & 2.7k & 20.4k & 56.3k \\
     % \hline
    \end{tabular}
\caption{\small Time comparison for different $\varepsilon$ for the double integrator experiment. The Reach-LP algorithm is run with GSG as the partitioner and CROWN as the propagator.}
\label{tab:double integrator}
\end{table}

\begin{figure}[t]
    \centering
    \vspace{2mm}
    \includegraphics[width=0.95\columnwidth]{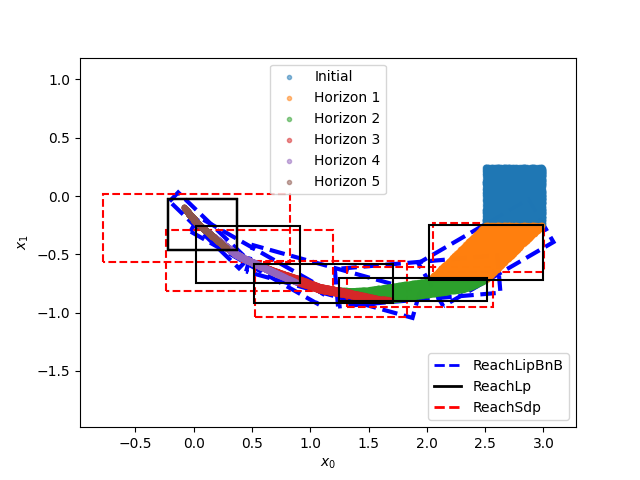}
    \caption{\small Double Integrator reachability analysis ($\varepsilon = 0.01$).}
    \label{fig:doubleIntegrator}
\end{figure}

\subsection{Double Integrator}
    We consider the discrete-time double integrator system from \cite{hu2020reach}, which can be written in the form of \eqref{eq:LTI system} with $A = 
    \begin{bmatrix}
    1 & 1\\
    0 & 1
    \end{bmatrix}$,
    $B = 
    \begin{bmatrix}
    0.5\\
    1
    \end{bmatrix}$.
    The control policy is a neural network that has been trained to approximate an MPC controller. %In \cite{everett2021efficient}, an MPC is used to generate data to train a three layer neural network to approximate the behaviour of the MPC.
    We train a neural network that has an architecture of $2 \times 10 \times 5 \times 1$ and run reachability analysis on this closed-loop system for five time steps and compare the results with 
    \cite{everett2021efficient, hu2020reach} in Fig. \ref{fig:doubleIntegrator}. 
    The run time statistics of this experiment in Tables [\ref{tab:double integrator}, \ref{tab:virtualBranching}], are calculated over 100 runs of the algorithm reported as $\mu \pm \sigma$.    
    Table \ref{tab:double integrator} shows the results for different $\varepsilon$ tolerances.

    \subsubsection{Bound Refinement}
    We test the effect of bound refinement ($k_v = 4$) on the performance of the algorithm. As reflected in Table \ref{tab:virtualBranching}, bound refinement increases the run time of the algorithm but significantly reduces the number of total branches.
    
    \begin{table}[b!]
\setlength{\tabcolsep}{2pt}
\centering
\resizebox{.7\columnwidth}{!}{%
    \begin{tabular}{|c | c | c c c |} 
     \hline
      \multicolumn{2}{|c|}{$\varepsilon$}  & 0.1 & 0.01 & 0.001 \\ 
     \hline
     \multirow{2}{*}{OR}& Run time[s] & 0.989 $\pm$ 0.058  & 1.204 $\pm$ 0.055 & 1.518 $\pm$ 0.066 \\ 
     
     & Num total branches & 1.1k& 4.3k & 8.8k \\
     \hline
      \multirow{2}{*}{BR}& Run time[s] & 1.061 $\pm$ 0.056 & 1.574 $\pm$ 0.069& 2.384 $\pm$ 0.068  \\ 
     
     & Num total branches & 0.5k& 2.3k& 5.2k\\
     \hline
      \multirow{2}{*}{LL}& Run time[s] & 0.949 $\pm$ 0.051 & 1.064 $\pm$ 0.055 & 1.226 $\pm$ 0.058 \\ 
     
     & Num total branches & 0.2k & 1.5k & 2.9k\\
     \hline
    \end{tabular}%
    }
\caption{\small Time and memory comparison for different $\varepsilon$ for the double integrator between the original (OR), bound refinement (BR), and lower Lipschitz (LL). The total number of optimization problems solved is 20.}
\label{tab:virtualBranching}
\end{table}

\subsubsection{Low Lipchitz Network}
To show the effect of the Lipschitz constant on our reachability task, we train a network with a lower Lipchitz constant using the method proposed in \cite{gouk2021regularisation} and report the results in Table \ref{tab:virtualBranching}. A lower Lipschitz constant would yield better lower bounds, hence allowing us to \texttt{Prune} more partitions, limiting the search space and reducing the overall run time.

% \begin{table}[b!]
% \setlength{\tabcolsep}{6pt}
% \centering
%     \begin{tabular}{|c | c | c | c| } 
%     \hline
%     $\varepsilon$ & 0.1 & 0.01 & 0.001\\
%     \hline
%      Run time[s] & & & \\
%      \hline
%       Number of total branches & 1 $\pm$ 2 & &\\ 
%       \hline
%     \end{tabular}
% \caption{\small Time and memory comparison of our algorithm for different sub-optimality tolerances for the double integrator experiment.}
% \label{tab:lowLip}
% \end{table}

\begin{figure}[t]
     \centering
     % \vspace{2em}
     \vspace{2mm}
     \includegraphics[width=0.95\columnwidth]{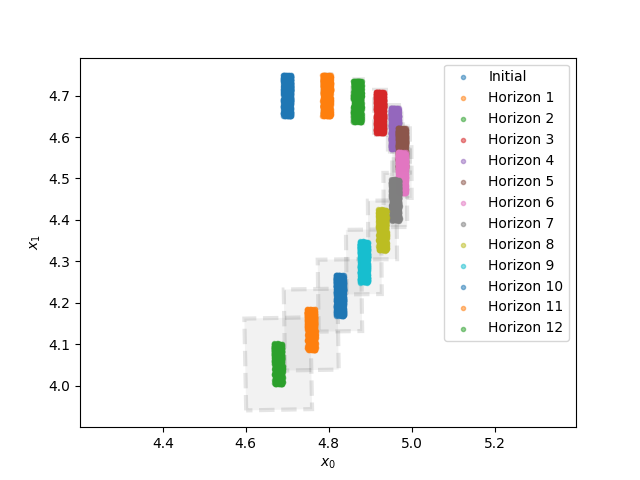}
     \caption{\small 6D Quadrotor reachability analysis. The grey area shows the over-approximations calculated by our method.}
     \label{fig:Quadrotor}
\end{figure}

\subsection{6D Quadrotor}
The 6D Quadrotor can be described using \eqref{eq:neural autonomous system} with 
$A = I_{6\times6} + \!\Delta t \times \begin{bmatrix}
0_{3\times3} & I_{3\times3} \\
0_{3\times3} & 0_{3\times3}
\end{bmatrix}$, 
$B = \!\Delta t \!\times  \!\begin{bmatrix}
& g & 0 & 0 \\
0_{3\times3} & 0 & -g & 0 \\
& 0 & 0 & 1
\end{bmatrix}^\top$, 
$u = \begin{bmatrix}
\tan(\theta) \\
\tan(\phi)\\
\tau
\end{bmatrix}$ and 
$c = \Delta t \times \begin{bmatrix}
0_{5\times1} \\
-g
\end{bmatrix}$. Similar to the previous test case, an MPC controller is designed for the discretized closed-loop system (with $\Delta t = 0.1$), and then a neural network is trained to approximate it. The neural network has an architecture of $6 \times 32 \times 32 \times3$.
Fig. \ref{fig:Quadrotor} shows the result of reachability analysis from an initial rectangle $ [4.69, 4.71]\times [4.65, 4.75] \times [2.975, 3.025]\times [0.9499, 0.9501]\times [-0.0001, 0.0001]\times [-0.0001, 0.0001]$ for 12 time steps (equivalent to $1.2$s in the continuous time system).

\section{Conclusion}
In this paper we proposed a branch-and-bound method for reachability analysis of neural networks, in which the bounding sub-routine is based on Lipschitz continuity arguments. The proposed algorithm is particularly efficient when the Lipschitz constant of the neural network is small.
%As it can be inferred from Eq. \eqref{eqn:lowerBound}, bounding is directly affected by the value of Lipschitz, therefore, this algorithm is well suited for problems with low Lipschitz constants. 
Several heuristics \cite{gouk2021regularisation, srivastava2014dropout} have been proposed to reduce the Lipschitz constant of neural networks during training.
For future work, we will explore more effective heuristics to improve the overall performance of the algorithm. 
%Hardware implementation of this algorithm could be considered one of the future research directions.

\printbibliography

\ifx
\clearpage
\appendix
\subsection{Proof of Proposition \ref{thm:boundsOnNewCoordinates}}
\begin{proof}
Let $x \in F(\overline{\mathcal{X}}^t)$. As $R^t$ is orthonormal, there exists a unique $y$ such that $x = R^{t + 1}y$, so
\[
y = {R^{t + 1}}^\top x = \begin{bmatrix}
{r^{t + 1, 1}}^\top x, &
\cdots,&
{r^{t + 1, n_x}}^\top x
\end{bmatrix}^\top.
\]
By definition in \eqref{eq: bounds on new coordinates}, we have 
\[
\ell^{t + 1}_i =
\min_{y \in [\ell^t,u^t]} \{{r^{t+1, i}}^\top F^t(R^t y)\} \leq 
{r^{t + 1, i}}^\top x = y_i,
\]

and
\begin{align*}
    \begin{split}
        y_i &= {r^{t + 1, i}}^\top x\! \leq \!
\max_{y \in [\ell^t,u^t]} \{{r^{t+1, i}}^\top F^t(R^t y)\}\! \\
\rightarrow y_i& \leq 
-\!\min_{y \in [\ell^t,u^t]} \{{-r^{t+1, i}}^\top F^t(R^t y)\} \!=\! u_i^{t + 1}. 
    \end{split}
\end{align*}
\vspace{-10.0mm}
\end{proof}
\fi
% \addtolength{\textheight}{-12cm}   % This command serves to balance the column lengths
% on the last page of the document manually. It shortens
% the textheight of the last page by a suitable amount.
% This command does not take effect until the next page
% so it should come on the page before the last. Make
% sure that you do not shorten the textheight too much.
%%%%%%%%%%%%%%%%%%%%%%%%%%%%%%%%%%%%%%%%%%%%%%%%%%%%%%%%%%%%%%%%%%%%%%%%%%%%%%%%

%%%%%%%%%%%%%%%%%%%%%%%%%%%%%%%%%%%%%%%%%%%%%%%%%%%%%%%%%%%%%%%%%%%%%%%%%%%%%%%%

%%%%%%%%%%%%%%%%%%%%%%%%%%%%%%%%%%%%%%%%%%%%%%%%%%%%%%%%%%%%%%%%%%%%%%%%%%%%%%%%

%%%%%%%%%%%%%%%%%%%%%%%%%%%%%%%%%%%%%%%%%%%%%%%%%%%%%%%%%%%%%%%%%%%%%%%%%%%%%%%%

\end{document}